\documentclass[12pt]{amsart}
\usepackage{graphicx}
\usepackage{amsaddr}
\usepackage{hyperref}
\usepackage{amsmath,amssymb}
\large\normalsize

\theoremstyle{plain}
\newtheorem{theorem}{Theorem}[section]
\newtheorem{lemma}[theorem]{Lemma}

\theoremstyle{definition}

\numberwithin{equation}{section}
\numberwithin{equation}{section}

\newcommand{\F}{\mathbb{F}}
\newcommand{\C}{\mathbb{C}}
\newcommand{\R}{\mathbb{R}}
\newcommand{\N}{\mathbb{N}}
\newcommand{\Z}{\mathbb{Z}}

\newtheorem{prop}[theorem]{Proposition}
\newtheorem{defi}[theorem]{Definition}

\newcommand{\ip}[2]
{\ensuremath{\langle #1,#2 \rangle}}

\usepackage{amsmath}

\begin{document}

	\title{Titchmarsh-Weyl theory for vector-valued discrete Schr\"odinger operators}
	\author{Keshav Raj Acharya}
	\address{Department of Mathematics, Embry--Riddle Aeronautical University\\ Daytona Beach, FL 32114-3900, U.S.A.\\
		acharyak@erau.edu}

	\maketitle
	\thispagestyle{empty}
	
	{\footnotesize
		\noindent{\bf Abstract:}		We develop the Titchmarsh-Weyl theory for vector-valued discrete Schr\"odinger operators and show that the Weyl $m$ functions associated with these operators map complex upper half plane to the Siegel upper half space. We also discuss about the Weyl disk and Weyl circle corresponding to these operators.\\

		\noindent{\bf Key Words}: Discrete Schr\"odinger operator, Titchmarsh-Weyl $m$-function.\\
		
		\noindent{ \bf AMS (MOS) Subject Classification:} 39A70, 47A05, 34B20.

	\maketitle

	\section{Introcuction}

	The goal of this paper is to extend the Titchmarsh-Weyl theory for vector valued discrete Schr\"odinger operators.   We consider a  discrete Schr\"odinger equation in $d-$ dimensional space of the form	
	\begin{align}\label{ds} y(n+1)+y(n-1)+B(n)y(n)=zy(n),\, \, z\in \C \end{align}
	where   $ y(n) =[y_1(n)\,\,y_2(n), \hdots\,\, y_d(n)]^t $ ( $t$ stands for a transpose),  is a vector valued sequence in $l^2(I, \C^d).$  Usually $I=\Z$ or $I=\N$. Here $l^2(I, \C^d)$ is a Hilbert space of   square summable vector valued sequences with the inner product \[\ip{u}{v} = \sum_{n\in I} u(n)^*v(n), \] where  $ `` \ast "$ stands for conjugate transpose  and $ B(n) $ is a symmetric $d\times d$ matrix. We denote the space of all $d\times d$ complex matrices by $\C^{d\times d}$. 	
	The equation (\ref{ds}) can  be generalized to a $d-$dimensional Jacobi equation of the form \begin{align} \label{je}A(n) y(n+1)+A(n-1)y(n-1)+B(n)y(n)=zy(n),\, \, z\in \C \end{align} with $A(n), B(n)$ are sequences of $d\times d$ matrices. If $I = N $ The equation (\ref{je}) can be written in the form: \begin{align*} \begin{pmatrix}	&B(1)&A(1)& 0 &   &\\
	& A(1)& B(2) & A(2) &  \ddots         &          \\
	 & 0 & A(2) & B(3) &  \ddots   &                \\
	  &     &\ddots & \ddots &\ddots &          \\
	\end{pmatrix} \begin{pmatrix} y(1)\\y(2)\\ \vdots \\ \vdots \\ \vdots	\end{pmatrix} = z \begin{pmatrix} y(1)\\y(2)\\ \vdots \\ \vdots \\ \vdots	\end{pmatrix}.\end{align*}	
	
	The matrix \begin{align*} J= \begin{pmatrix}	&B(1)&A(1)& 0 &   &\\
	& A(1)& B(2) & A(2) &  \ddots         &          \\
	& 0 & A(2) & B(3) &  \ddots   &                \\
	&     &\ddots & \ddots &\ddots &          \\
	\end{pmatrix} \end{align*} is called a block Jacobi matrix. Some studies about the block Jacobi matrix can be found in the paper \cite{RK}. Equation (\ref{ds}) is a particular case of Jacobi equation with  $A(n)\equiv 1.$ 
	
  The equation (\ref{ds}) induces a discrete Schr\"odinger operator $J$ on $ l^2(I, \C^d)$ as \[
	\mathop{J} y(n)= y(n+1)+y(n-1)+B(n)y(n).\] It can be easily observed that if $B(n)$  is a Hermitian matrix, $ B(n)^*=B(n),$ then $J$ is a self-adjoint operator on $ l^2(\N, \C^d)$. Then, the spectrum of $J$ is a set of real numbers: $\sigma(J) \subset \R .$ \\ 
	
To get a solution of the equation (\ref{ds}), we may fix any two vectors $c, d \in \C^d$ at two consecutive sites, that is, we  fix the values  $u(k) = c, u(k+1) = d$ and evolve according to (\ref{ds}). In particular, we fix $u(0)$ and $u(1)$ then any $u(n)$ is obtained by solving the difference equation (\ref{ds}) using transfor matrices: \begin{align} T(m;z) = \begin{pmatrix} zI-B(m) & - I \\ I & 0 \end{pmatrix} \end{align}  where $I$ is an $d\times d$ identity matrix. Let  \begin{align} A(n;z) =T(n;z)\times \dots \times T(1,z) \times I  .\end{align}
Then,  $u$ solves (\ref{ds}) for every $n$ if and only if \begin{align} \label{tm} \begin{pmatrix}  u(n+1)\\ u(n)  \end{pmatrix}= A(n;z) \begin{pmatrix}  u(1)\\ u(0)  \end{pmatrix}  \end{align}	

This matrix can also be used to get a solution at cite $n$ from cite $m$ as \[\begin{pmatrix}  u(n+1)\\ u(n)  \end{pmatrix}= A(n, m;z) \begin{pmatrix}  u(m+1)\\ u(m)  \end{pmatrix} . \]
  For every pair of vectors $c, d \in \C^d,$ there exists a solution of (\ref{ds}), therefore, the space of solutions of (\ref{ds}) is a $2d $-dimensional vector space. In \cite{KA}, it is shown that are exactly $d$ linearly independent solutions of (\ref{ds}) that are in $l^2(\N, \C^d).$ 
	
	It is now convenient to fix a basis of the solution space of (\ref{ds}). An easier way   is to prescribe 	a pair of initial conditions. For $z\in \C,$ let  \begin{align}\label{is} U(n,z) = (u_1(n), u_2(n), \hdots, u_d(n)), \hspace{.5in} V(n,z) = (v_1(n), v_2(n), \hdots, v_d(n)) \end{align} where $  u_i(n)= [u_{1,i}(n)\,\,u_{2,i}(n)\,\, \hdots u_{d,i}(n)]^t  \hspace{.5in} v_i(n)= [v_{1,i}(n)\,\,v_{2,i}(n)\,\, \hdots v_{d,i}(n)]^t$ are solutions of (\ref{ds}). Thus, both of the sets $U(n,z)$ and $V(n,z)$ consists of $d$ linearly independent solutions of  $(\tau -z)u(n)= 0$, where $\tau$ is the expression on the left side of equation (\ref{ds}). For our convenience, we call these sets as  matrix valued solutions  of (\ref{ds}).  We further suppose that these solutions satisfy the following initial conditions \begin{equation} \label{ic} U(0,z)= -I, \hspace{.3in} V(0,z)= 0, \hspace{.2in}U(1,z) =0, \hspace{.3in} V(1,z)= I .\end{equation}
	By iterating the difference equation, we see that for fixed $n\in \N, U(n,z),  V(n,z)$ are polynomial of degree $n-2$ over $ \C^{d\times d}.$  So $ \overline{U(n,z)} = U(n,\bar{z})$ and $ \overline{V(n,z)}=V(n,\bar{z}).$
	
	We generalize the equation (\ref{tm}) for the matrix valued solutions $ U(n,z), V(n,z)$ as  \begin{align*}  \begin{pmatrix}  U(n+1,z) & V(n+1,z)\\ U(n,z) & V(n, z) \end{pmatrix}  & = A(n;z) \begin{pmatrix}  U(1, z) & V(1, z)\\ U(0,z) & V(0, z)  \end{pmatrix}\\ &  = A(n;z) \begin{pmatrix}  0 & I \\ -I & 0  \end{pmatrix}\\ & = A(n;z) \mathbb J , \end{align*} where $ \mathbb J = \begin{pmatrix}  0 & I \\ -I & 0  \end{pmatrix}$

\begin{lemma}\label{W} Suppose $n \in \N_0= \N \cup \{0\},$ and  $W(z)= \begin{pmatrix}  U(n+1,z) & V(n+1,z)\\ U(n,z) & V(n, z) \end{pmatrix}$  then \[ W^t\mathbb JW = W\mathbb JW^t = \mathbb J \]\end{lemma}	
	\begin{proof} Notice that $T(n;z)^t\mathbb J T(n;z) = T(n;z)\mathbb J T(n;z)^t =\mathbb J $ for any $n$ so that
		
		   $A(n;z)^t\mathbb J A(n;z) = A(n;z)\mathbb J A(n;z)^t = \mathbb J .$ Then \begin{align*} W^t\mathbb JW  & = ( A(n;z) \mathbb J)^t \mathbb J A(n;z)\\ & = \mathbb J^t A(n;z)^t \mathbb J A(n;z)\mathbb J \\ & = \mathbb J^t\mathbb J\mathbb J \\ & = \mathbb J \end{align*} 
		
			Exactly the same way we can see: $W\mathbb JW^t = \mathbb J $ \end{proof}

			\begin{defi} The Wronskian of any two sequences  $f(n,z), g(n,z) \in l^2(\N, \C^d) $ is defined by \[ W_n (f,g)= [ f^*(n+1, \bar{z}) g(n,z) -  f^* (n, \bar{z})g(n+1,z)] .\] \end{defi} 
			
			This definition incorporate with the definition in one dimensional space and  in the continuous case. In \cite{KA}, it is shown that for fixed $z\in \C$,  if  $f(n,z), g(n,z) \in l^2(\N, \C^d) $ are any two solutions of (\ref{ds}) then $W_n (f, g)$ is independent of $n.$ Moreover, the Wronskian $W_n$ is linear in both arguments.
			
		 For $f(n,z), g(n,z) \in l^2(\N_0, \C^d)$ the Green's identity corresponding to equation (\ref{ds}) is given by  \[ \sum _ {j=0}^n \Big( f^* (\tau g) - (\tau f)^* g \Big)(j) = W_0( \bar{f},g)- W_n(\bar{f}, g) .\]

			We extend the definition of  Wronskian and the Green's identity  for the matrix valued solutions $U(n,z), V(n,z) $, each contains $d$ linearly independent solutions of (\ref{ds}) for fixed $z\in \C$.  \[ W_n (U,V)= [ U^*(n+1, \bar{z}) V(n,z) -  U^* (n, \bar{z})V(n+1,z)] .\] 
			
			It is shown in \cite{KA} that the Wronskian  $  W_n (U,V)$ is a matrix independent of $n\in\N.$  We extend the Green's Identity for these matrix valued solutions.
			\begin{align}\label{gi}  \sum_{j=0}^N \Big( F(j,z)^*(\tau G(j,z)) - (\tau F(j,z))^* G(j,z) \Big) = W_0(\bar{F}, G)-  W_N(\bar{F}, G).\end{align} Again, the proof of the Green's identity can be found in \cite{KA}.

	\section{Titchmarsh-Weyl $m$ function}
	
	The theory of Titchmarsh-Weyl $m$ functions is very important tool in the spectral theory of Jacobi and Schr\"odinger operators. In order to study the asymptotic behavior of solutions of Jacobi and Schr\"odinger equations, one need to study these $m$ functions. Moreover,  the absolutely continuous, singular continuous and essential spectrum  of the operators associated with these equations are well explained   in terms of $m$ functions. These $m$ functions were first introduced in 1910 by H. Weyl  in \cite{WM} for Sturn-Liouville differential equations. It was further studied by E. C. Titchmarsh in \cite{EC} and established the connection between the analyticity of the solution and the spectrum of the operator of Sturn-Liouville differential equations. For further history of $m$ function, one can see \cite{E}. The  theory of $m$ functions in one dimensional space has been widely studied, some of  which can be found in the papers \cite{BR, SF, FG, RC, BS, GT}. 
	
The Titchmarsh-Weyl $m$ function for the vector-valued discrete Schr\"odinger operators associated to the equation (\ref{ds}) is defined in terms of solutions as follows. 
	
	\begin{defi} Let $z\in \C^+= \{z\in \C : \operatorname{Im}(z) >0\}.$ The Titchmarsh-Weyl $m$ function is defined as the unique complex matrix $ M(z) \in \C^{d\times d}$ such that \begin{align} \label{wm} F(n,z) = U(n,z)+ V(n,z)M(z) \end{align} where $U(n,z), V(n,z) $ are matrix valued solutions consisting of $d$ linearly independent solutions with initial values (\ref{ic})  and the matrix valued solution $  F(n,z)$ is a set of $d$ linearly independent solutions of  (\ref{ds}) that are in $l^2(\N, \C^d).$\end{defi}	
	
This definition, is in fact well defined. As we mentioned above that there are only $d$ linearly independent solutions in $l^2(\N_0, \C^d)$, if there is another $M(z)$ satisfying the above conditions then the solutions from both  $U(n,z)$ and $V(n,z)$ will be in $l^2(\N_0, \C^d)$. The solution   $V(n,z)$ is such that $V(0,z)= 0$ which implies that  $V(n,z)$ is the set of eigen-functions for the self adjoint operator $J.$ This contradicts that the spectrum of $J$ is a set of real numbers.

	\begin{theorem}\cite{KA} Let $z\in \C^+.$ If $(\tau-z)F= 0$  and $F $ is a $d\times d$ matrix valued solution whose $d$ columns are linearly independent solutions of (\ref{ds}) that are in $ l^2(\N, \C^d).$ Then \begin{align} \label{mf} M(z)= -F(1,z)F(0,z)^{-1} .\end{align} Moreover, \begin{align} \label {mf2}M(z) =(m_{ij}(z))_{d\times d} \in \C^{d\times d}, \,\, m_{ij}(z) =\ip{\delta_j}{(J-z)^{-1}\delta_i}.\end{align} \end{theorem}
	
	\begin{proof} If the matrix valued soulution $F$ is given by (\ref{wm}) then $F(0,z)=-I$ and $F(1,z)=M(z).$ So (\ref{mf}) holds. Suppose $G(n,z)$ is any $d\times d$ matrix valued solution then it is a constant (matrix) multiple of the solution set $F(n,z)$ from (\ref{wm}) because (\ref{wm}) is a set of $d$ linearly independent solutions.  That is, \begin{align*} G(n,z)= F(n,z)C  \end{align*} where $C$ is a $d\times d$ scalar invertible matrix.  \begin{align*}F(n,z)
		= G(n,z)C^{-1} \end{align*} so that
		\begin{align*} -G(1,z)G(0,z)^{-1} =  & -F(1,z)CC^{-1}F(0,z)^{-1} \\ = &  -F(1,z)F(0,z)^{-1} \\ = & M(z) .\end{align*} Let $F(n,z)$ as in (\ref{mf}) and let \[ g_i= (J-z)^{-1} \delta_i\] where $\delta_i(n) \in  l^2(\N, \C^d)$ such that the values of  $\delta_i(n)=0$ if $ i\neq 0 $ and  $\delta_i(i)= [1,0,\dots 0]^t .$  Then $(J-z)g_i = \delta_i. $ So $(\tau-z)g_i(n) =0$ for $n\geq 2.$ Moreover $g_i \in l^2 $ for all $i=1,2,.......,d.$ Let $$G(n,z)= [g_1, g_2,......., g_d].$$ Then $G(n,z)= F(n,z)C,\,\, C\in \C^{d\times d}.$
		By comparing values at  $$n=1, \,\, G(1,z)= [g_1(1), g_2(1), ... ... ..., g_d(1)].$$ Here $$g_1(1)= (J-z)^{-1} \delta_1(1)$$ and $$g_1=[g_{11}, g_{21},...,...,     ...,g_{d1}]^t,\,\, g_{i1}= \ip{\delta_i}{ g_1} , i= 1,2,....,d.$$
		
		Then $M(z) = G(1,z)C^{-1}$ and  \begin{align*}M(z) & = (m_{ij}(z)) \\ & = ( \ip{\delta_j}{ (J-z)^{-1}\delta_i})C^{-1}. \end{align*} To find the value of $C,$ we compare values at $n=2.$
		
		First $(J-z)G(1,z)= (\delta_1, \delta_2,... ... , \delta_d)$ so 
		$$ (J-z)G(1,z)= \begin{pmatrix} 1 & 0 \hdots     &0\\ 0 &1  \hdots    &0\\ \vdots& \vdots  & \vdots\\ 0 & 0      \hdots   &1  \end{pmatrix} = I $$ It follows that $$ G(2,z)+ B(1)G(1,z)-zG(1,z)= I $$ $$G(2,z) = (z-B(1)) G(1,z)+I.............(i)$$ 
		Also,  	
		$$ F(2,z)= (z-B(1)) F(1,z)-F(0,z)C $$ $$ G(2,z) = (z-B(1))G(1,z)- G(0,z) ........ ....(ii) $$\\ Comparing (i) and (ii), we get	$-F(0,z)C=I$ and so $I.C=I \implies C=I .$ Hence (\ref{mf2}) holds. That is \begin{align} \nonumber M(z) & = ( m_{ij}(z)) \\ & = ( \ip{\delta_j }{(J-z)^{-1}\delta_i} ) .\end{align}
		
		\end{proof}

This result allows us to connect the $m$ function with a matrix valued Borel measure using functional calculus for these resolvent operators $ \ip{\delta_j }{(J-z)^{-1}\delta_i},$ where $\delta_i(n) \in  l^2(\N, \C^d)$ such that the values of  $\delta_i(n)=0$ if $ i\neq 0 $ and  $\delta_i(i)= [1,0,\dots 0]^t $

 	 By functional calculus,\[  m_{ij}(z) = \int_{\R} \frac{1}{t-z} d\mu_{ij}\]
 	 where $\mu_{ij}$ is a spectral measure for the vectors $\delta_j$ and $\delta_i$. Therefore, \[ M(z) =\int_{\R} \frac{1}{t-z} d\mu, \,\,\, \mu= (\mu_{ij})_{d\times d}\] and  \[ M(z) =\int_{\R} \frac{1}{t-z} d\mu = \Big( \int_{\R} \frac{1}{t-z} d\mu_{ij} \Big)_{d\times d} \]
 	 The matrix valued measure $\mu$ is a spectral measure of the $d-$dimensional discrete Schr\"odinger  operator $J.$
 	 
 	 For each $i, j$ the entries $m_{i,j}(z)$ maps complex upper half plane to itself. For if $z \in \C^+$,  $ \operatorname{Im} m_{ij}(z)= \frac{1}{2i}(  m_{ij}(z) - m_{ij}( \bar{z})) = \int_{\R} \frac{y}{|t-z|^2}d\mu_{ij}  > 0 $

 	 Suppose $\overline{M(z)}$ denotes the complex conjugate of $M(z)$ obtained by taking the complex conjugate of each entries of $ M(z) $. Then by integral representation of $ m_{ij}(z)$, we have $ m_{ij}(z) = m_{ij}( \bar{z}))$ so that   $ \overline{M(z)} = M(\bar{z})$.\\ Also,  $M(z)=( m_{ij}(z)) = ( \ip{\delta_j }{(J-z)^{-1}\delta_i} ) $ so that
 	 	\begin{align*} m_{ij}(z) = & \ip{\delta_j }{(J-z)^{-1}\delta_i} \\ = & \ip{ (J-\bar{z})(J-\bar{z})^{-1}\delta_j}{(J-z)^{-1}\delta_i} \\ = &  \ip{ (J-\bar{z})^{-1}\delta_j}{(J-\bar{z})^*(J-z)^{-1}\delta_i}  \end{align*} Since $J$ is self adjoint,  $(J-\bar{z})^*= (J-z)$	\begin{align*}m_{ij}(z)= &  \ip{ (J-\bar{z})^{-1}\delta_j}{\delta_i}  \\ = & \overline{\ip{\delta_i}{(J-\bar{z})^{-1}\delta_j} }\\	 
 	 	= & \overline{m_{ji}(\bar{z})}\\
 	 	= & m_{ji}(z) \end{align*} for all $i, j.$ Hence $M(z)^t = M(z).$ Thus we proved the following proposition.
 	 	
 	 	\begin{prop} $M(z)^*= M(\bar{z}), $\end{prop}
 	 	
 The imaginary part of $M(z)$ is  $ \operatorname {Im}M(z) = \frac{1}{2i}(M(z)- M(z)^*) $ and it is clear from the above observation that $\operatorname {Im}M(z) > 0 .$
 	 
Let $\mathcal S$ be a subspace of $\C^{d\times d}$, consisting of all symmetric matrices with positive definite imaginary part. That is, \[ \mathcal S = \{ M \in \C^{d\times d} : \frac{1}{2i}(M - M^*) >0 \}\]	 
 The space $\mathcal S $ is called a Seigel upper half space.
 
 From above discussion we proved\\
 
 \begin{theorem} For $z\in \C^+,$ the map $z \mapsto M(z)$ maps complex upper half plane $\C^+$ to Seigel upper half space $\mathcal S $. \end{theorem}

 \section {Titchmarsh-Weyl circles and disks}
 
 In this section, we  define the Titchmarsh-Weyl circles and disks. We consider the equation (\ref{ds}) on  a compact interval $[0, N].$ Suppose 
 $U(n,z), V(n,z) $ are the matrix valued solutions of (\ref{ds}) with initial values (\ref{ic}).  For $z\in \C^+$, define a matrix valued solution 
 
 $F(n,z) = U(n,z)+ V(n,z) M_N^{\beta}(z) $ satisfying a boundary condition
 
 $ \beta_2 F(N,z) + \beta_1 F(N+1, z)  = 0 $ where \begin{align} \label{bc}  \beta =[\beta_1, \beta_2 ]\in \R^{d\times 2d},\, \beta_1, \beta_2 \in \R^{d\times d},  \,\,  \beta ^t \beta = I, \,\,\,\beta J  \beta^t = 0 . \end{align}
 The unique coefficient $ M_N^{\beta}(z) $ is called the Weyl $m$ function on the interval $[0, N].$
 
  On solving we see that, \begin{align} \label{M} M_N^{\beta}(z) & = -  \big(\beta_2 V(N,z)  + \beta_1 V(N+1, z)  \big)^{-1 }\big(\beta_2 U(N,z)  +\beta_1  U(N+1, z)  \big) .\end{align} Note that $ \big(\beta_2 V(N,z)  + \beta_1 V(N+1, z)  \big)$ is invertible.  Since $z, N. \beta $ varies, $ M_N^{\beta}(z)$ becomes a function of these arguments, and since $ U, V $ are matrix polynomials with entries meromorphic functions of $z$.

 \begin{lemma}\label{wms} The weyl $m$ function $M_N^{\beta}(z)$ on $[0, N]$ is symmetric. \end{lemma}
 
 \begin{proof} Let $\mathbb U(z) =\begin{pmatrix}  U(N+1)\\ U(N)  \end{pmatrix} = A(N;z) \begin{pmatrix}  U(1)\\ U(0)  \end{pmatrix} = A(N;z) \begin{pmatrix}  0\\ -I  \end{pmatrix} $
 	  and  
 	  
 	   $ \mathbb V(z) =\begin{pmatrix}  V(N+1) \\ V(N)  \end{pmatrix} = A(N;z) \begin{pmatrix}  V(1)\\ V(0)  \end{pmatrix} = A(N;z) \begin{pmatrix}  I \\ 0 \end{pmatrix}$\\
 	   
Using equation (\ref{M}), the Weyl $m$ function can be written as $ M_N^{\beta}(z) = - (\beta \mathbb V(z))^{-1}(\beta\mathbb U(z) ) .$ Suppose $E =\beta\mathbb U(z) $	 and  $F =\beta\mathbb V(z) $ so that $M_N^{\beta}(z)= - F^{-1}E .$
Now, \begin{align*} M_N^{\beta}(z)^T -  M_N^{\beta}(z) & = F^{-1}E- (F^{-1}E)^T \\ &= F^{-1} [ FE^T -EF^T] F^{-T}  \\ &=  F^{-1} [\beta \mathbb V (\beta \mathbb U)^T - \beta \mathbb U (\beta \mathbb V)^T ] F^{-T} \\ &= F^{-1}\beta [\mathbb V \mathbb U^T - \mathbb U \mathbb V^T ]\beta^T F^{-T} \\ & =  F^{-1}\beta\Big[A(N;z) \begin{pmatrix}  I \\ 0 \end{pmatrix} \Big( A(N;z) \begin{pmatrix}  0 \\ -I \end{pmatrix} \Big)^T   \\ &  - A(N;z) \begin{pmatrix}  I \\ 0 \end{pmatrix}\Big( A(N;z) \begin{pmatrix}  0 \\ -I \end{pmatrix} \Big)^T \Big]\beta^T F^{-T}  \\ &= - F^{-1}\beta\Big[A(N;z) J  A(N;z) ^T  \Big]\beta^T F^{-T} \\ &= -  F^{-1}\beta J \beta^T F^{-T} \\ &= 0 \end{align*}   

\end{proof}

 \begin{lemma} \label{pdm} For a matrix valued solution   $F(n,z) = U(n,z)+ M_N^{\beta}(z) V(n,z) $ of (\ref{ds}) we have $W_N(\bar{F}, F) = 2i\operatorname{Im} M - 2i\operatorname{Im}z \sum_{j=0}^N F(j,z)^* F(j,z).$ 
 \end{lemma}
  
\begin{proof} We use the Greens identity (\ref{gi})	 with $G=F.$   
 \begin{align*}  \sum_{j=0}^N \Big( F(j,z)^*(\tau F(j,z)) - (\tau F(j,z))^* F(j,z) \Big) = W_0(\bar{F}, F)-  W_N(\bar{F}, F) \\ (z-\bar{z}) \sum_{j=0}^N F(j,z)^* F(j,z) = W_0(\bar{F}, F)-  W_N(\bar{F}, F) \end{align*}   
 	For  $F(n,z) = U(n,z)+ M_N^{\beta}(z) V(n,z) $, using the linearity of the  Wronskian we get
 	 \begin{align*} & \sum_{j=0}^N \Big( F(j,z)^*(\tau F(j,z)) - (\tau F(j,z))^* F(j,z) \Big) = W_0(\bar{F}, F)-  W_N(\bar{F}, F) \\ &(z-\bar{z}) \sum_{j=0}^N F(j,z)^* F(j,z) = W_0(\bar{F}, F)-  W_N(\bar{F}, F) \\ & =  W_0(\overline{U+VM}, U+VM)-  W_N(\bar{F}, F) \\& =  W_0(\overline{U}, U)+  W_0(\overline{U}, VM) +  W_0(\overline{VM}, U) +  W_0(\overline{VM}, VM ) - W_N(\bar{F}, F) .\end{align*}  Then we have

 	    Here   $W_0(\overline{U}, U) =  W_0(\overline{VM}, VM ) =0, \,\,\, W_0(\overline{VM}, U)= -\bar{M},\,\,W_0(\overline{U}, VM) = M$
 
  \begin{align*} & (z-\bar{z}) \sum_{j=0}^N F(j,z)^* F(j,z)   = M - \bar{M}  - W_N(\bar{F}, F) \\ & 2i\operatorname{Im}z \sum_{j=0}^N F(j,z)^* F(j,z) = 2i\operatorname{Im} M -  W_N(\bar{F}, F) \\ & W_N(\bar{F}, F) = 2i\operatorname{Im} M - 2i\operatorname{Im}z \sum_{j=0}^N F(j,z)^* F(j,z) \end{align*} \end{proof}
 The condition on $ \beta$ in the boundary condition (\ref{bc}) implies that $\beta_1$ and $\beta_2$ are invertible. 
 Equation (\ref{M}) is written as \begin{align*}  M_N^{\beta}(z) & = -  \big(\beta_2 V(N,z)  + \beta_1 V(N+1, z)  \big)^{-1 }\big(\beta_2 U(N,z)  +\beta_1  U(N+1, z)  \big)  \\ & = -  \big( \beta_1 ^{-1}\beta_2 V(N,z)  +  V(N+1, z)  \big)^{-1 }\big( \beta_1 ^{-1} \beta_2 U(N,z)  + U(N+1, z)  \big)    \\ &= -  \big(\gamma V(N,z)  +  V(N+1, z)  \big)^{-1 }\big( \gamma U(N,z)  + U(N+1, z)  \big)  ,\,\,\, \gamma =  \beta_1 ^{-1} \beta_2   \in \R^{d\times d}. \end{align*}

 Again solving for $ \gamma$ we have,\[  \gamma = -F(N+1,z)F(N,z)^{-1} .\]

 Observe  that $\Im \gamma = \frac{1}{2i}(\gamma - \gamma^*) = 0 $

 Let $ \mathcal W(N, z, M ) = \begin{pmatrix} U(N+1, z) & V(N+1,z) \\ U(N, z) & V(N,z) \end{pmatrix}\begin{bmatrix}I\\ M \end{bmatrix} .$ Define a matrix function\[ E(M, N) = -i \mathcal W(N, z, M )^*J \mathcal W(N, z, M ) \]
 
 Observe that  \begin{align} \label{wd1} \nonumber E(M, N)  & = -i [F(N+1,z)^*, F(N,z)^*] J \begin{bmatrix}F(N+1,z)\\F(N,z)\end{bmatrix} \\  \nonumber &= -i W_N(\bar{F}, F ) \\  & = -2\operatorname{Im} M + 2\operatorname{Im}z \sum_{j=0}^N F(j,z)^* F(j,z).\end{align}

 \begin{defi}  Let $z \in \C^+ .$ The set \[ \mathcal D(N,z) = \{ M\in C^{d\times d}| E(M, N) \leq 0 \} \text {  and } C(N,z) = \{ M \in C^{d\times d}| E(M, N) = 0  \} \] are respectively called the Weyl disk and Weyl circle.  \end{defi}
 
  Clearly,  $C_N(z) = \{M_N^{\beta}(z): \beta \in \R^{d\times d}, \text{  satisfying  } (\ref{bc}) \}.$  Thus for any complex symmetric matrix $M \in \C^{d\times d}$ \[ M \in C(N,z) \iff \operatorname{ Im }(- F(N+1,z) F(N, z)^{-1} ) = 0  \]
 
 \begin{theorem} The map $ z\mapsto M_N^{\beta}(z)$ maps complex upper half plane to Seigel half space.     \end{theorem}
 	
 	\begin{proof}By lemma \ref{wms}$M_N^{\beta}(z)$ is symmetric.   Since $M_N^{\beta}(z) \in C(N,z), $ $E(M, N) =0.$ It follows that $ -i W_N(\bar{F}, F ) =0 .$ By lemma \ref{pdm} we have \[ 2\operatorname{Im} M - 2\operatorname{Im}z \sum_{j=0}^N F(j,z)^* F(j,z) =0 .\] That is \[ \frac{\operatorname{Im} M }{\operatorname{Im}z} =  \sum_{j=0}^N F(j,z)^* F(j,z) > 0  .\] This implies that $ \operatorname{Im} M $ is positive definite. \end{proof}
 	
 	\begin{lemma}[Nesting property of Weyl disks] Let $z \in \C^+ .$ Then \[\mathcal D(N+1,z) \subset \mathcal D(N,z),\,\, N \in \N\] \end{lemma}
\begin{proof} Let $M\in \mathcal D(N+1,z). $ From (\ref{wd1}) we have
	\begin{align*} E(M, N) & = -2\operatorname{Im} M + 2\operatorname{Im}z \sum_{j=0}^N F(j,z)^* F(j,z) \\ & \leq -2\operatorname{Im} M + 2\operatorname{Im}z \sum_{j=0}^{N+1} F(j,z)^* F(j,z) \\ & = E(M, N+1) \leq 0 .\end{align*} This shows that $M \in \mathcal D(N,z).$ Hence the result. \end{proof}
From above we have,
\begin{align*} E(M, N)  & = -i [I, M^*] \begin{pmatrix} U(N+1, z)^* & U(N, z)^* \\ V(N+1,z)^* & V(N,z)^* \end{pmatrix} J \begin{pmatrix}  U(N+1, z) & V(N+1,z) \\ U(N, z) & V(N,z) \end{pmatrix}\begin{bmatrix} I\\ M \end{bmatrix} \\ &=  -i [I, M^*]\begin{pmatrix} W_N(\bar{U}, U)& W_N(\bar{U}, V) \\ W_N(\bar{V}, U) & W_N(\bar{V}, V) \end{pmatrix}\begin{bmatrix} I\\ M \end{bmatrix} \\ &= -i [ W_N(\bar{U}, U)+ W_N(\bar{U}, V)M + M^*W_N(\bar{V}, U) + M^* W_N(\bar{V}, V)M   ] \end{align*}

Using $ W_N(\bar{V}, V)^* = -W_N(\bar{V}, V)$ and $  W_N(\bar{V}, U)^* =-W_N(\bar{U}, V) $,  $ E(M, N)$ can be written as

\begin{align} \label{eq 4}E(M, N)   = -i \Big\{ [M-  W_N(\bar{V}, V)^{-1} W_N(\bar{U}, V)^* ]^*  W_N(\bar{V}, V) [ M - W_N(\bar{V}, V)^{-1}W_N(\bar{U}, V)^* ]\\ \nonumber + W_N(\bar{U}, U) + W_N(\bar{U}, V)W_N(\bar{V}, V)^{-1} W_N(\bar{U}, V)^*  \Big\}\end{align}

\begin{lemma}  \label{lemma 6}For $z\in \C^+, W_N(\bar{U}, V)W_N(\bar{V}, V)^{-1} W_N(\bar{U}, V)^* +  W_N(\bar{U}, U) = - W_N(V, \bar{V})^{-1} $ \end{lemma}

\begin{proof} Let $\mathbb W = W^*J W  .$ Notice that $ W^*J W = \begin{pmatrix} W_N(\bar{U}, U)& W_N(\bar{U}, V) \\ W_N(\bar{V}, U) & W_N(\bar{V}, V) \end{pmatrix} .$  From lemma \ref{W} we see that \[ W^*J\bar{W} = J .\] Then  \begin{align*}\mathbb W^tJ  \mathbb W  & =  (W^*JW)^t J (W^*JW) \\ &=  W^tJ^t W^{*^t} J W^* J W   \\ &=- W^tJ W  \\ &= J .\end{align*} 
On the other hand,\begin{align*}\mathbb W^tJ  \mathbb W  & = \begin{pmatrix} W_N(\bar{U}, U)^t & W_N(\bar{V}, U)^t \\  W_N(\bar{U}, V)^t   & W_N(\bar{V}, V)^t \end{pmatrix} J \begin{pmatrix} W_N(\bar{U}, U)& W_N(\bar{U}, V) \\ W_N(\bar{V}, U) & W_N(\bar{V}, V) \end{pmatrix} 
\\ & = \begin{pmatrix} W_N( U, \bar{U})^* & W_N(V, \bar{U})^* \\  W_N(U, \bar{V})^*   &  W_N(V, \bar{V})^* \end{pmatrix}  J \begin{pmatrix} W_N(\bar{U}, U)& W_N(\bar{U}, V) \\ W_N(\bar{V}, U) & W_N(\bar{V}, V) \end{pmatrix}
\end{align*}
By direct computation   we see that 
\begin{align} \label{1} -W_N(V,\bar{V})^*  W_N(\bar{U}, U )+W_N(U,\bar{V})^*W_N(\bar{V}, U)= -I \\ \label{2}- W_N(V, \bar{V})^*  W_N(\bar{U}, V)+  W_N(U, \bar{V})^* W_N(\bar{V}, V)=0 .\end{align}
From  equation (\ref{2}) we have \begin{align*} W_N(\bar{U}, V)^* & = W_N(\bar{V},V)^*W_N(U,\bar{V})W_N(V,\bar{V})^{-1} \\ & = -W_N(\bar{V},V)W_N(U,\bar{V})W_N(V,\bar{V})^{-1}  \end{align*}

Using this we get 
\begin{align*}  & W_N(\bar{U}, V)W_N(\bar{V}, V)^{-1} W_N(\bar{U}, V)^* +  W_N(\bar{U}, U) \hspace{2in} \\ &  = W_N(\bar{U}, V)W_N(\bar{V}, V)^{-1}[-W_N(\bar{V},V)W_N(U,\bar{V})W_N(V,\bar{V})^{-1}  ] + W_N(\bar{U}, U) \\ & = -W_N(\bar{U}, V)W_N(U,\bar{V}) W_N(V,\bar{V})^{-1}  +  W_N(\bar{U}, U) \end{align*} Also from equation (\ref{1}) we get,
\begin{align*}W_N(U,\bar{V})^*W_N(\bar{V}, U) = -I + W_N(V,\bar{V})^*  W_N(\bar{U}, U )\\ -W_N(U,\bar{V})^*W_N(\bar{U}, V)^* = -I + W_N(V,\bar{V})^*  W_N(\bar{U}, U )\\ (W_N(\bar{U}, V)W_N(U,\bar{V}))^*=I - W_N(V,\bar{V})^*  W_N(\bar{U}, U )\\ W_N(\bar{U}, V)W_N(U,\bar{V})=I -  W_N(\bar{U}, U )^*W_N(V,\bar{V})\\ W_N(\bar{U}, V)W_N(U,\bar{V})= I + W_N(\bar{U}, U )W_N(V,\bar{V}) \end{align*}

Then, 
\begin{align*}  & W_N(\bar{U}, V)W_N(\bar{V}, V)^{-1} W_N(\bar{U}, V)^* +  W_N(\bar{U}, U) \hspace{2in} \\& = -(I + W_N(\bar{U}, U )W_N(V,\bar{V})) W_N(V,\bar{V})^{-1}  +  W_N(\bar{U}, U) \\ & = - W_N(V,\bar{V})^{-1} . \end{align*} \end{proof} 

Using 
Using lemma (\ref{lemma 6}) and equation (\ref{eq 4}) we can express $E(M, N)$ in the form 
\begin{align*}  E(M, N)= -i \Big\{ [M-  W_N(\bar{V}, V)^{-1} W_N(\bar{U}, V)^* ]^*  W_N(\bar{V}, V)\\ [ M - W_N(\bar{V}, V)^{-1}W_N(\bar{U}, V)^* ] - W_N(V,\bar{V})^{-1} \Big\}. \end{align*}   Thus it can be expressed as \begin{align*}  E(M, N)= - [(M - C_N(z))^* R(N,z)^{-2}(M -C_N(z)) - R(N,\bar{z})^{2} ]\end{align*} where $ C_N(z) = W_N(\bar{V}, V)^{-1} W_N(\bar{U}, V)^* $ and $ R(N,z)= ( i W_N(\bar{V}, V))^{-1/2} .$
 So the equation of Weyl circle can be written as\[ (M - C_N(z))^* R(N,z)^{-2}(M -C_N(z)) =  R(N, \bar{z} )^{2} \]

\begin{theorem}For all $z \in \C^+, $ $ \lim_{N\rightarrow \infty }R(N,z) $ exists and $ \lim_{N\rightarrow \infty } R(N,z) \geq 0 .$ \end{theorem}

\begin{proof} By Green's identity we have \begin{align*} 2 \operatorname{Im}z  \sum_{j=0}^N V(j, z)^* V(j,z) = iW_N(\bar{V}, V) =  R(N,z)^{-2} > 0 . \end{align*} Also $ R(N,z)^{-2}$ is nondecreasing. Thus, $ R(N,z)$ is non increasing and so $ \lim_{N\rightarrow \infty }R(N,z) $ exists. \end{proof}

\begin{theorem}For all $z \in \C^+, $ $ \lim_{N\rightarrow \infty }C_N(z) $ exists.  \end{theorem}

\begin{proof}For any $M \in C(N, Z) $  we have
\[ (M - C_N(z))^* R(N,z)^{-2}(M -C_N(z)) =  R(N, \bar{z} )^{2} \] which follows that 
 \[ \Big(R(N, z )^{-1} (M - C_N(z)) R(N,\bar{z})^{-1}\Big)^* \Big( R(N,z)^{-1}(M -C_N(z))R(N, \bar{z} )^{-1}\Big) = I  \]
 Suppose $ U = \Big(R(N, z )^{-1} (M - C_N(z)) R(N, \bar{z})^{-1}\Big) $ so that $U^*U = I$ that is $U$ is unitary. Also, \[M = C_N(z) + R(N, z )U  R(N, \bar{z}) \]
 Suppose $ M \in C_{N+1}(z) \subset C_N(z) $ then we have \\
 $M = C_{N+1}(z) + R(N+1, z )U_{N+1}  R(N+1, \bar{z}) $ and  
 $M = C_N(z) + R(N, z )U_N  R(N, \bar{z}) $ Equating and taking operator norm on both sides we get \begin{align*} \| C_{N+1}(z)- C_N(z)\| & = \|R(N+1, z )U_{N+1}  R(N+1, \bar{z})-  R(N, z )U_N  R(N, \bar{z}) \|\\ & \leq  \|R(N+1, z )U_{N+1}  R(N+1, \bar{z}) - R(N, z )U_{N+1}  R(N+1, \bar{z})\| \\ & +  \|R(N, z )U_{N+1}  R(N+1, \bar{z}) -  R(N, z )U_{N}  R(N+1, \bar{z})\| \\ & + \|R(N, z )U_{N}  R(N+1, \bar{z}) -  R(N, z )U_N  R(N, \bar{z}) \| \\ & \leq  \|R(N+1, z ) - R(N, z ) \| \|U_{N+1}\| \|  R(N+1, \bar{z})\| \\ & +  \|R(N, z )\|U_{N+1}  -  U_{N} \| \| R(N+1, \bar{z})\| \\ & + \|R(N, z )\| \|U_{N}\| \|  R(N+1, \bar{z}) -  R(N, z ) \|\end{align*}
  This shows that $C_N(z)$ is a Cauchy sequence, hence converges.
\end{proof}
 
Let $C_0(z) = \lim_{N\rightarrow \infty} C_N(z)$ and $ R_0(z)= \lim_{N\rightarrow \infty} R(N, z)$.\\ Define $ D_0(z) = \{M\in \C^{d\times d}: (M-C_0(z))^*R_0(z)^{-2}(M-C_0(z)) \leq R_0(\bar{z})^2$ then \begin{align*}D_0(z) = \cap_{N \geq 1} D(N, z). \end{align*}
 
 \begin{theorem} Let $z\in \C^+$ and $M \in \C^{d\times d}.$ Then for $F(N,z) = U(N,z)+V(N,z)M$ we have\\
 	(1) $M$ is inside $\mathcal D_0(z)$ if and only if \[ \sum_{N=1}^{\infty} F(N,z)^*F(N,z) \leq \frac{\operatorname{Im} M}{\operatorname{Im} z}\]
 	(2)  $M$ is on the boundary of $\mathcal D_0(z)$ if and only if \[ \sum_{N=1}^{\infty} F(N,z)^*F(N,z) = \frac{\operatorname{Im} M}{\operatorname{Im} z}\] \end{theorem}

 \begin{proof} 
 	Let $M \in \mathcal D_0(z). $ Then $M \in D(N,z)$ for all $N.$ So from (\ref{wd1}) we have 	\begin{align*} E(M, N) & = -2\operatorname{Im} M + 2\operatorname{Im}z \sum_{j=0}^N F(j,z)^* F(j,z) \leq 0\end{align*}
 which follows that \begin{align*} \sum_{j=0}^N F(j,z)^* F(j,z) \leq \frac{\operatorname{Im} M }{\operatorname{Im}z} .\end{align*} Taking limit as $N\rightarrow \infty$ we get  \[ \sum_{N=1}^{\infty} F(N,z)^*F(N,z) \leq \frac{\operatorname{Im} M}{\operatorname{Im} z}.\]	
 Conversely, for any $N$ we have, \begin{align*} \sum_{j=1}^N F(j,z)^* F(j,z) \leq \sum_{j=1}^{\infty} F(j,z)^* F(j,z)  \leq \frac{\operatorname{Im} M }{\operatorname{Im}z} .\end{align*} So $ E(M, N) \leq 0 $ for all $N$ and hence $M \in \mathcal D_0(z).$ Similar explanation also proves (2).	\end{proof}

 	\textbf{Acknowledgement:} The author would like to thank the Department of Mathematics and the Office of Sponsored Research, Embry-Riddle Aeronautical University for support.




	\providecommand{\bysame}{\leavevmode\hbox to3em{\hrulefill}\thinspace}
	\providecommand{\MR}{\relax\ifhmode\unskip\space\fi MR }
	\providecommand{\MRhref}[2]{%
		\href{http://www.ams.org/mathscinet-getitem?mr=#1}{#2}
	}
	\providecommand{\href}[2]{#2}

	\end{document}